\documentclass[letterpaper,11pt,twoside]{article}
\usepackage[left=1in, right=1in, bottom=1.25in, top=1.5in]{geometry}
\usepackage[utf8]{inputenc}
\usepackage{setspace} 
\usepackage{fancyhdr}
\usepackage{xcolor}
\usepackage[mathscr]{euscript}

\usepackage{amsmath,amssymb,amsthm,amsfonts}
\usepackage{latexsym,bbm,xspace,graphicx,float,mathtools,mathdots,xspace}
\usepackage{thmtools, thm-restate}
\usepackage{enumitem}
\usepackage[backref, colorlinks,citecolor=blue,linkcolor=magenta,bookmarks=true]{hyperref}
\usepackage[nameinlink]{cleveref}
\usepackage[normalem]{ulem}

\newtheorem{theorem}{Theorem}

\newtheorem{lemma}{Lemma}

\newtheorem{defn}{Definition}

\newtheorem{claim}{Claim}

\newtheorem{property}{Property}

\def\colorful{1}
\ifnum\colorful=1

\fi
\ifnum\colorful=0

\fi

\fancypagestyle{plain}{%
\fancyhf{} % clear all header and footer fields
\fancyfoot[C]{\textbf{\thepage}} % except the center

}

\title{A Lower Bound on Cycle-Finding in Sparse Digraphs}
\author{Xi Chen, Tim Randolph, Rocco A.~Servedio, Timothy Sun 
\\ \texttt{ \{xichen, rocco, tim\}@cs.columbia.edu, t.randolph@columbia.edu }
\\ Columbia University}
\date{}

\newcommand{\br}{Bender-Ron}
\newcommand{\BR}{\mathbb{BR}}
\newcommand{\CH}{\mathbb{S}}

\def\BKG{\mathsf{BKG}}

\newcommand{\eps}{\epsilon}
\newcommand{\simple}{\mathrm{simple}}
\newcommand{\Prx}{\mathop{{\bf Pr}\/}}
\newcommand{\UB}{{Q}}
\newcommand{\UBQ}{{Q^*}}
\def\bS{\boldsymbol{S}}
\newcommand{\bX}{\boldsymbol{X}}
\newcommand{\bY}{\boldsymbol{Y}}
\newcommand{\bG}{\boldsymbol{G}}
\newcommand{\bC}{\boldsymbol{C}}

\newcommand{\KG}{\mathsf{KG}}
\newcommand{\VKG}{\mathsf{VKG}}

\newcommand{\alga}{\mathcal{A}}

\newcommand{\MinDegree}{80}
\newcommand{\MaxEpsilon}{1/60}

\def\calT{\mathcal{T}}
\def\sb{\text{blue}}
\def\sr{\text{red}}
\def\sh{\mathsf{h}}
\def\calE{\mathcal{E}}

\def\bB{\mathbf{B}}
\def\bR{\mathbf{R}}

\def\calT{\mathcal{T}}
\def\sw{\mathsf{w}}
\def\bU{\mathbb{U}}
\def\bbU{\mathbb{U}}
\def\anc{\mathsf{anc}}
\def\calA{\mathcal{A}}

\begin{document}

\maketitle
\begin{abstract}

We consider the problem of finding a cycle in a sparse directed graph $G$ that is promised to be far from acyclic, meaning that the smallest feedback arc set in $G$ is large. We prove an information-theoretic lower bound, showing that for $N$-vertex graphs with constant outdegree any algorithm for this problem must make $\tilde{\Omega}(N^{5/9})$ queries to an adjacency list representation of $G$. In the language of property testing, our result is an $\tilde{\Omega}(N^{5/9})$ lower bound on the query complexity of one-sided algorithms for testing whether sparse digraphs with constant outdegree are far from acyclic. This is the first improvement on the $\Omega(\sqrt{N})$ lower bound, implicit in Bender and Ron \cite{BR02}, which follows from a simple birthday paradox argument.
\end{abstract}

\thispagestyle{empty}

\newpage
\setcounter{page}{1}

\section{Introduction\vspace{-0.06cm}}

In the current massive data era there is great interest in the abilities and limita\-tions of \mbox{sublinear} time algorithms for various computational problems.  In particular, in recent years a number of~researchers have studied sublinear time algorithms for fundamental graph problems such as approximating  the size of the minimum vertex cover \cite{PR07,MR09,NO08,YYI09,HKNO09,ORRR12}, the number of connected components \cite{CRT05}, maximum matching \cite{NO08,YYI09}, and the minimum spanning tree weight \cite{CRT05, CS09,CEFMNRS05}; counting 
edges \cite{Feige06,GR08,Beame18}, stars \cite{Gonen11,Aliakbarpour2018}, triangles  \cite{ELRS17}, $k$-cliques \cite{Eden18}, and arbitrary subgraphs \cite{AKK18}; finding forbidden minors \cite{KSS18,KSS19}; and checking $k$-colorability \cite{RodlDuke85}, bipartiteness \cite{GoldreichRon02}, planarity \cite{BSS08}, and more. The sublinear time regime imposes natural constraints on algorithms. For instance, a simple ``needle in a haystack'' lower bound argument shows that it is impossible~to distinguish acyclic graphs from graphs with one or more cycles in sublinear time. As a result, sublinear graph algorithms typically provide either approximate guarantees on their output\footnote{For instance, the algorithms of  \cite{ORRR12, CRT05, ELRS17,AKK18}.} or are designed for property testing-style problems in which the input graph $G$ is promised to satisfy some condition that allows a sublinear algorithm to succeed.\footnote{For instance, the algorithms of \cite{KSS18,RodlDuke85,GoldreichRon02,BSS08}.} Our results are of the second type: \emph{We prove a lower bound on the running time of algorithms for finding cycles in sparse digraphs that are promised to be not too close to acyclic.}

To motivate our inquiry, we observe that many of the most fascinating and enigmatic objects of modern scientific research, such as brains, neural networks, social networks, and the Internet, are naturally modeled as \emph{massive, sparse, directed graphs}. Thus it is a compelling goal to understand the capabilities of sublinear time algorithms on such graphs. Despite this fact, although there is a substantial literature on property testing in general undirected graphs (see Chapters 8, 9, and 10 of \cite{Goldreich17book}), we are aware of fewer works on sublinear time algorithms or property testing on sparse directed graphs \cite{OR11,HS12,HS13,CPS16}.  The most directly relevant previous work that we are aware of is the early paper of Bender and Ron \cite{BR02} on testing acyclicity in directed graphs, which we discuss in detail below.\vspace{-0.1cm}

\subsection{The query model and promise problem that we consider\vspace{-0.05cm}}

Throughout this work, we consider digraphs on $N$ vertices named $[N]:=\{1,\ldots,N\}$ in which the outdegree of each vertex~is bounded from above by a small absolute constant $d$. It suffices to take $d \geq \MinDegree$ for our main result to hold. Graphs are represented using the adjacency list model, in which a query consists of a vertex $u \in [N]$ and an index $i \in [d]$. In response, the algorithm receives the $i^{th}$ outneighbor of $u$ or an empty string if $u$ has fewer than $i$ outneighbors. The query complexity of an algorithm is the maximum number of queries that it makes on any $N$-vertex graph.

The algorithmic problem we consider is that of outputting a directed cycle given an input graph $G$. The promise that $G$ contains at least one cycle is insufficient to allow sublinear time algorithms: for instance, if $G$ consists of a single constant-length cycle and all other vertices are isolated, or if $G$ consists of a single cycle of length $N$, $\Omega(N)$ queries are required. Hence, in the spirit of property testing, we consider the promise problem in which the input graph $G$ is promised to be \emph{$\eps$-far from acyclic}. This means that the smallest \emph{feedback arc set}\footnote{Recall that a subset $S \subset E$ of directed edges in a graph is a feedback arc set if every directed cycle in $G$ contains at least one edge in $S$, or equivalently, deleting all the edges in $S$ makes $G$ become acyclic.} of $G$ is of size at least $\eps d N$.
 
As we discuss later in item (1) of \Cref{sec:future}, the promise that $G$ is $\eps$-far from acyclic ensures that $G$ must contain very short cycles \cite{Fox18}, so this promise eliminates the concern that outputting a directed cycle will already necessitate $\Omega(N)$ runtime.  However, it is far from clear how many queries may be required to \emph{find} a cycle in sparse directed graphs that are $\eps$-far from acyclic.  This question was implicitly considered by Bender and Ron: in \cite{BR02} they gave an $\Omega(N^{1/3})$-query lower bound on property testing algorithms for testing whether a bounded-degree digraph is acyclic versus $\eps$-far from acyclic with two-sided error in the adjacency list model.  Implicit in the proof of their $\Omega(N^{1/3})$ lower bound is an $\Omega(N^{1/2})$ lower bound for one-sided testers, or equivalently, for algorithms that find a directed cycle in far-from-acyclic bounded-degree digraphs.  We give a proof sketch of this lower bound in \Cref{sec:techniques-BR}; as we explain there, their $\Omega(N^{1/2})$ lower bound is based on a simple birthday paradox argument but such an argument cannot succeed in obtaining an $\omega(N^{1/2})$ lower bound.  We note that Bender and Ron \cite{BR02} state as an explicit goal for future work the problem of improving their lower bound, and that acyclicity testing in bounded-degree digraphs is listed as ``Open Problem \#41'' on the website \texttt{sublinear.info}.\footnote{\href{https://sublinear.info/index.php?title=Open_Problems:4}{\texttt{https://sublinear.info/index.php?title=Open\_Problems:41}}}

\subsection{Our result:  An $\tilde{\Omega}(N^{5/9})$-query lower bound.}  
Our main result is a proof that any randomized algorithm under the 
  adjacency list  query model must make $\tilde{\Omega}(N^{5/9})$ queries to find a cycle in a sparse $N$-vertex digraph that is $\eps$-far from acyclic. The lower bound holds even if $\eps$ is a fixed constant. 
In more detail, our main result is the following:

\begin{restatable}[Main theorem]{theorem}{main}
\label{thm:main}
Let $d,\eps$ be fixed constants with $d \geq \MinDegree$ and $\eps \le \MaxEpsilon$, and let $G$ be an arbitrary digraph, promised to be $\eps$-far from acyclic and with outdegree bounded above by $d$. Any algorithm that, given query access to the adjacency list representation of $G$, outputs a directed cycle in $G$ with constant probability must make $\tilde{\Omega}(N^{5/9})$ queries.
\end{restatable}

We give a detailed discussion of our techniques in \Cref{sec:techniques}; as explained there,  the arguments underlying our lower bound are significantly more involved, both conceptually and technically, than the $\Omega(N^{1/2})$ lower bound of \cite{BR02} for the same problem. 

\section{Preliminaries}

A \emph{directed graph} (or \emph{digraph}) $G = (V, E)$ consists of a set $V$ of vertices and a set $E$ of directed edges. Each edge directed from $u$ to $v$ is represented by the pair $(u,v)$. The \emph{outdegree} (resp. \emph{indegree}) of a vertex $u$ is the number of edges $(u,v)$ (resp. $(v,u)$) between $u$ and an \emph{outneighbor} (resp. \emph{inneighbor})~$v \in V$. We say a digraph has outdegree bounded by $d$ if every vertex has outdegree at most $d$. A digraph is \emph{$\eps$-far from acyclic} if the minimum feedback arc set has size $\eps d N$ (that is, at least $\eps d  N$ edges must be removed to make $G$ acyclic). An \emph{out-tree} is an acyclic digraph in which there exists a unique directed path from a \emph{root} vertex to every other vertex. 
A vertex in an out-tree with no outgoing edge is called a \emph{leaf}; otherwise it is called an \emph{internal vertex}. An out-tree is said to have degree $d$ if every internal vertex has outdegree exactly $d$.

Given a positive integer $n$, we write $[n]$ to denote $\{1,\ldots,n\}$. For fixed $d$ and $\eps$, we consider the problem of finding a cycle in a digraph $G=([N],E)$, with outdegree bounded by $d$, that is $\eps$-far from acyclic. Algorithms may query the \emph{adjacency list representation} of $G$ as follows. We assume the algorithm knows $N$. A query consists of a vertex $u \in [N]$ and an index $i \in [d]$. In response, the algorithm receives the $i^{th}$ outneighbor of $u$ or an empty string if $u$ has fewer than $i$ neighbors. For convenience, we simplify the adjacency list query model to the $\emph{vertex query}$ model, in which the algorithm simply queries a vertex $u$ and receives an ordered list containing all outneighbors~of $u$. Clearly, algorithms on digraphs with maximum outdegree at most $d$ under the vertex query model can be implemented in the adjacency list model by increasing the number of queries by a factor of $d$, and thus asymptotic lower bounds in the vertex query model also hold in the adjacency list model. 
   
\section{Our techniques} \label{sec:techniques}

As is standard in property testing, we employ Yao's principle \cite{Yao:77} to prove our lower bound.  By this principle, to prove \Cref{thm:main} it suffices to define a probability distribution over $N$-vertex digraphs with outdegree
  bounded by $d$ and argue that 
\begin{enumerate}
    \item A random $\bG$ drawn from this distribution is $\eps$-far from acyclic with probability $1-o_N(1)$.
    \item Any deterministic algorithm $\alga$ that makes 
    $$Q^*:=\frac{N^{5/9}}{\log N}$$ queries to $\bG$ finds a cycle with probability $o_N(1)$. 
\end{enumerate}

In this section we first present a simple distribution from \cite{BR02} and sketch the $\Omega(N^{1/2})$ lower bound for this distribution that is implicit in the arguments of \cite{BR02}. We then outline the difficulty inherent in proving an
  asymptotically better lower bound, informally describe the distribution $\BR$ that we use for \Cref{thm:main}, and outline our proof of the theorem.

\subsection{A simple $\Omega(N^{1/2})$ lower bound due to Bender and Ron} \label{sec:techniques-BR}

The distribution over sparse digraphs we now describe corresponds to the distribution  $\mathscr{G}_2$ defined in Section~4 of \cite{BR02}; we denote this distribution by $\BR_{\simple} := \BR_{\simple}(N, d).$ A graph $\bG$ drawn from $\BR_{\simple}$ is generated by randomly partitioning the $N$ vertices $\{1,\dots,N\}$ into two equal-size subsets $\bS_1$ and $\bS_2$ and taking $d$ random directed matchings from $\bS_1$ to $\bS_2$ and $d$ random directed matchings from $\bS_2$ to $\bS_1$ as the edges of $\bG.$  

A straightforward probabilistic analysis (see Lemma~5 of \cite{BR02}) shows that for any constant $d \geq 128$, a random graph $\bG \sim \BR_{\simple}$ is $\eps$-far from acyclic for $\eps=1/16$. To complete the lower bound, it remains to argue that any deterministic algorithm $\calA$ that makes $o(N^{1/2})$ queries finds~a directed cycle in $\bG \sim \BR_{\simple}$ with probability $o_N(1)$. This follows from the following stronger property: with probability $1-o_N(1)$ over the choice of $\bG \sim \BR_{\simple}$, no deterministic algorithm that makes $o(N^{1/2})$ queries receives in response to a query a vertex it has previously observed, either as input to or output from a query.\footnote{We assume without loss of generality that the algorithm never repeats a previous query.} This property follows from a standard birthday paradox type argument, i.e., the fact that a sequence of $o(N^{1/2})$ uniform samples from an $N$-element set samples the same element twice with probability $o_N(1)$.

\subsection{A challenge in going beyond $N^{1/2}$ many queries}

Another birthday paradox type argument demonstrates that a random walk in $\bG \sim \BR_{\simple}$ will collide with itself in $O(N^{1/2})$ steps with high probability, thus yielding a cycle. Hence a different construction must be considered to obtain an $\omega(N^{1/2})$ lower bound.

The essence of the simple $\Omega(N^{1/2})$ lower bound is that with high probability, each of $o(N^{1/2})$ many queries yields an answer that is uniform over all previously unseen vertices, in which case the algorithm receives \emph{no useful information} about the underlying graph. Unfortunately, this attractive property does \emph{not} hold for algorithms that make $\omega(N^{1/2})$ queries. For example, an~algorithm that repeatedly queries {$(u,i)$} pairs drawn uniformly at random from $[N]\times [d]$ would observe i.i.d. draws from some distribution over $[N]$. Because $\omega(N^{1/2})$ i.i.d. draws from \emph{any} distribution supported on at most $N$ elements will result in $\omega_N(1)$ collisions with high probability, any argument establishing an $\omega(N^{1/2})$ lower bound must contend with the nontrivial information that algorithms receive about the unknown underlying graph through collisions. Indeed, the central difficulty in proving an $\omega(N^{1/2})$ lower bound is showing that no algorithm can gain enough information from induced collisions to find a cycle.

\subsection{Our construction and a sketch of our main ideas} \label{sec:techniques-construction}

We now give an informal description of the distribution $\BR := \BR(N,d)$ that we analyze (a detailed description is given in \Cref{sec:procedure}).  This distribution is a modified version of a construction proposed by Bender and Ron in \cite{BR02}.

Each graph in the support of $\BR$ has $3N$ vertices, and each vertex has outdegree either $d$ or $0$. A graph $\bG$ drawn from $\BR$ is obtained as follows:  $N$ vertices are randomly selected and designated as \emph{blue} vertices, and the remaining $2N$ vertices are designated as \emph{red} vertices. Red vertices are randomly partitioned into $L$ many layers $R_1,\dots,R_L$, each containing $W=2N/L$ vertices.\footnote{Both $L$ and $W$ are $N^{\Theta(1)}$; the exact values will be given later and are not important for our intuitive discussion here.}  Each blue vertex is assigned $d$ outneighbors by choosing each one uniformly at random from the blue vertices and the first half of the layers of the red vertices.  Each red vertex in layer $R_i$ ($i < L$) is assigned $d$ outneighbors by choosing each one uniformly from the $W$ vertices in $R_{i+1}.$  For a visual example, refer to \Cref{fig:BR-graph}. A straightforward probabilistic argument (given in \Cref{sec:far-from-acyclic}) shows that with probability $1-o_N(1)$ a random $\bG \sim \BR$ is $\eps$-far from acyclic, so the main challenge is to show that it is hard to find a directed cycle in a graph drawn from this distribution.

\begin{figure}[h!]
    \centering
    \includegraphics[]{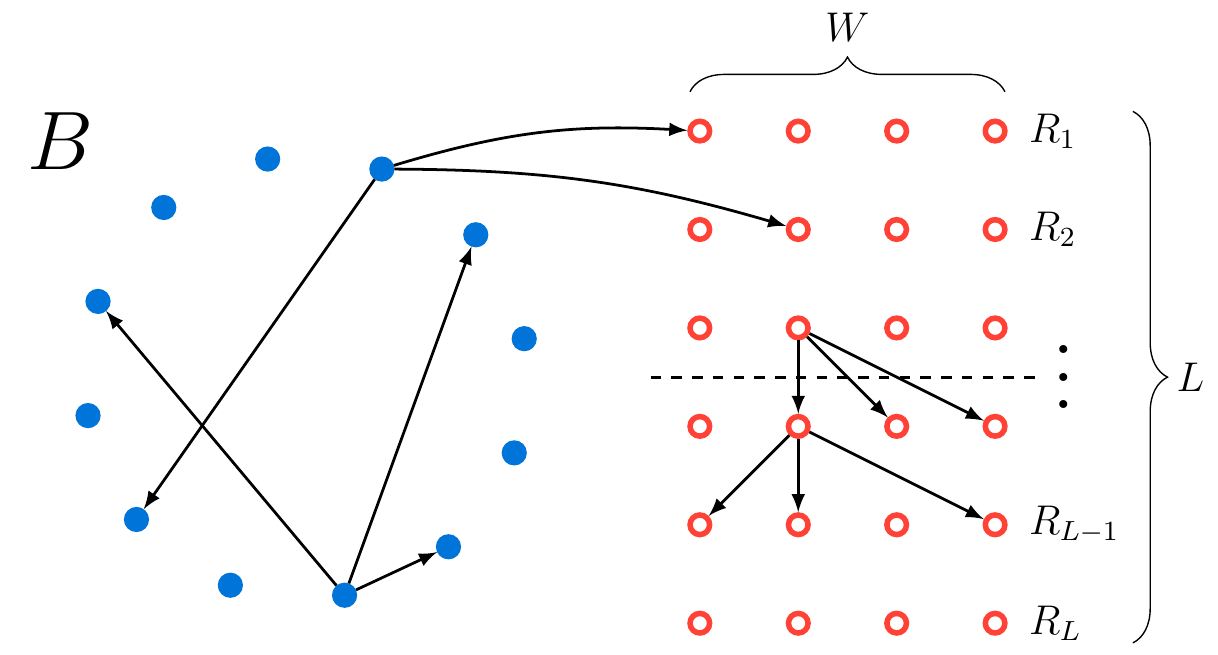}
    \caption{Cartoon of a random graph $\bG \sim \BR$.}
    \label{fig:BR-graph}
\end{figure}

We give some intuition behind the construction of graphs in $\BR$. Note that every cycle in $\bG$ consists entirely of blue vertices. Thus, a cycle-finding algorithm may want to ``avoid wandering into the red region.'' This, however, is difficult to do because the local neighborhood of a typical vertex ``looks the same'' whether it is blue or red (note that an algorithm under the adjacency list model of course never receives explicit information about whether any particular vertex is blue or red).  For example, the simple random walk approach sketched at the beginning of the previous subsection will not work for $\bG \sim \BR$: even if the random walk starts at a blue vertex, after $O(1)$ steps on average it will reach a red vertex and will have no chance of completing a cycle. Given that an algorithm needs $\Omega(N^{1/2})$ queries to find a cycle even if it is given the set of blue vertices (since the blue part of $\bG\sim\BR$ is very similar to graphs drawn from $\BR_{\simple}$ described in Section \ref{sec:techniques-BR}), it is natural to hope for an $\omega(N^{1/2})$ lower bound using the distribution $\BR$.  

There are two challenges in obtaining an $\omega(N^{1/2})$ lower bound using $\BR$. First, as discussed in the previous subsection (which applies not only to $\BR$ but to any distribution), an $\omega(N^{1/2})$-query algorithm may experience many collisions and hence potentially obtain a significant amount of information about $\bG$. The second challenge is specific to $\BR$. Despite the intuition, ``wandering into the red region'' may actually provide useful information about $\bG$ when done strategically (see \Cref{sec:sketch-algorithm} for two attacks on $\BR$ based on exploring the red region; they together imply that one cannot hope to obtain a lower bound better than $N^{13/18}$ using $\BR$). Given that many algorithmic strategies are possible, how can one argue that \emph{every} algorithm that does not make too many queries is unlikely to find a cycle?

To explain the intuition that underlies our lower bound, we first note that for a query on vertex $u$ to reveal a cycle, it must be the case that $u$ is blue and there is a directed path from one of its outneighbors to $u$ in the current ``knowledge graph'' of the algorithm (where the knowledge graph consists of all edges that have been found so far and $v$ is an ancestor of $u$ if it has a directed path to $u$). As a result, we focus on the maximum number of ancestors among all blue vertices in the current knowledge graph because the probability that an algorithm discovers a cycle when it queries a vertex is proportional to its number of ancestors in the knowledge graph. Our proof, at a high level, shows that this crucial quantity cannot grow too fast.

A key notion behind our analysis is the division of a sequence of queries made by an algorithm into distinct \emph{epochs}. Roughly speaking, an epoch ends either when a collision occurs (i.e., one of the outneighbors of the vertex queried is a vertex that the algorithm has seen before, either as a query vertex or as an outneighbor of a query vertex), or when ``too many'' queries have been made since the end of the previous epoch. We introduce the notion of epochs in \Cref{epochssec} and bound the number of epochs that occur in the execution of any algorithm that makes at most $Q^*$ queries (\Cref{lemma:epochbound}). We also bound the number of \emph{blue surprise epochs}: these are epochs that end because the vertex $u$ queried is blue and has a blue outneighbor $v$ that the algorithm has seen before (\Cref{lemma:blueepochbound}). We pay special attention to such epochs because with the discovery of $(u,v)$, all ancestors of $u$ become ancestors of $v$ and thus, the number of ancestors of $v$ may grow rapidly.

Next, in~\Cref{sec:no-long-blue-path} we show that during an epoch of any algorithm, regardless of outcomes of previous epochs, the vertices queried are unlikely to contain a path of blue vertices of length more than $4 \log N$. This is captured in \Cref{maintechnicallemma}, which is at the heart of our lower bound argument. In particular, \Cref{maintechnicallemma} implies that during an epoch that is not a blue epoch (or during a blue surprise epoch but ignoring the blue-blue collision edge found at the end of the epoch), the maximum number of ancestors of blue vertices in the knowledge graph can increase by no more than $4\log N$.

Finally, in~\Cref{sec:cycle-finding}, we combine Lemmas \ref{lemma:epochbound}, \ref{lemma:blueepochbound}, and \ref{maintechnicallemma}  to bound the maximum number of ancestors of blue vertices in its knowledge graph during the execution of any $\UB^*$-query algorithm on $\bG\sim \BR$.
This is used to show that every such algorithm finds a cycle in $\bG\sim \BR$ with probability $o_N(1)$. 
 
To simplify the presentation, we introduce in \Cref{sec:model} an augmented query model, called the \emph{color revelation model}, in which more information is provided to the algorithm than in the standard model.  Specifically, at the end of each epoch the query algorithm is provided with the color of every vertex it has previously seen. All our results discussed above are proved under this model and our lower bound trivially carries over to the standard model since any algorithm under the latter can be simulated under the color revelation model by simply ignoring the additional information.

\section{The Bender-Ron graphs} \label{sec:procedure}
 
In this section we formally describe the distribution $\BR := \BR(N,d)$ and prove, in \Cref{sec:far-from-acyclic}, that $\bG\sim \BR$ is $\MaxEpsilon$-far from acyclic with probability $1-o_N(1)$ when $d\ge \MinDegree$. \Cref{thm:main} follows from the next theorem which we prove in the rest of the paper.

\begin{theorem}\label{secondmaintheorem}
Let $d$ be a constant with $d\ge \MinDegree$.
Let $\alga$ be any $Q^*$-query deterministic algorithm that operates on graphs in the support of $\BR$
  under the vertex-query model, where $Q^* :=N^{5/9}/\log(N)$. 
Then the probability of $\alga$ finding a cycle in $\bG\sim \BR$ is  $ o_N(1)$.
\end{theorem}

\subsection{The distribution} \label{sec:thedist}

Let $W := 2N/L = (2N)^{7/9}$ and $L:=(2N)^{2/9}$ be two parameters indicating the width of each red layer and the number of red layers, respectively.\footnote{Note that by our choices of $L$ and $W$, $N+LW = 3N$. This particular setting of $L$ and $W$ is chosen to optimize our lower bound, as will become clear in the course of our analysis.  We  assume without loss of generality that $N$ is such that both $L/2$ and $W$ are integers.}
We refer to a map from a subset of $[3N]$ to $L+1$ colors $\{\sb,\sr_1,\ldots,\sr_L\}$
  as a \emph{coloring}.

A digraph $\bG\sim \BR$ over the vertex set $[3N]$ is generated by the following randomized procedure:
\begin{flushleft}\begin{enumerate}
    \item 
    Let $\bbU$ be the uniform distribution over all colorings 
      $C:[3N]\rightarrow \{\sb,\sr_1,\ldots,\sr_L\}$ such that 
      $N$ vertices are colored $\sb$ and $W$ vertices are colored $\sr_i$ for 
      each $i\in [L]$. The procedure starts by drawing a coloring $\bC\sim \bbU$.
Naturally we refer to vertices in $\bB$ as blue vertices and 
  vertices in $\bR_1\cup \cdots\cup \bR_L$ as red vertices in $\bC$.
We view $\bR_1,\ldots,\bR_L$ as $L$ layers of red vertices and refer
  to vertices in $\bR_i$ as red vertices in the $i^{th}$ layer (see \Cref{fig:BR-graph}).

\item For each blue vertex $u\in \bB$, 
  create its adjacency list by drawing a sequence of $d$
  vertices \emph{without} replacement from the following set of $(N-1)+LW/2=2N-1$ vertices:
    \begin{equation}
    \big(\bB\setminus \{u\}\big) \cup \bigcup_{i=1}^{L/2} \bR_i.
    \end{equation}
Thus, a blue vertex has $d$ distinct outneighbors from
  $\bB$ and the top $L/2$ layers of red vertices.
    \item For each red vertex in $\bR_i$, $1 \leq i < L$, create its adjacency list by drawing a sequence of $d$ vertices
    without replacement from $\bR_{i+1}$.
Thus, each red vertex (other than those in the bottom layer $\bR_L$)
  has $d$ distinct outneighbors in the next layer.
Finally, set the adjacency list of each vertex in $\bR_L$
  to be empty.
This finishes the construction of $\bG$. 
Note that every vertex in $\bG$ has out-degree either $d$ or $0$ so
  $\bG$ is a bounded-outdegree-$d$ digraph as promised.
  
\end{enumerate}\end{flushleft}
We refer to graphs in the support of $\BR$ as \br{} graphs, since these graphs are inspired by a construction that was proposed (but not analyzed) in \cite{BR02}.
\Cref{fig:BR-graph} illustrates a graph in $\mathbb{BR}$.
To facilitate our proof of \Cref{secondmaintheorem} later, 
  in addition we introduce $\BR^*$ to denote the distribution
  of $(\bC,\bG)$ generated by the procedure above (so the marginal distribution
  of $\bG$ in $\BR^*$ is the same as $\BR$).

We record the following property that is trivial from the construction:

\begin{property}\label{basicproperty}
Let $(C,G)$ be a pair in the support of $\BR^*$ and let
  $(u,v)$ be an edge in $G$.
Then either (1) $C(u)=C(v)=\sb$ (a $\sb \rightarrow\sb$ edge), (2) $C(u)=\sb$ and $C(v)=\sr_i$ for some $i\le L/2$
(a $\sb \rightarrow\sr$ edge), or (3) $C(u)=\sr_i$ and $C(v)=\sr_{i+1}$
  for some $i<L$ (a $\sr\rightarrow\sr$ edge).

Moreover, if a vertex $u$ has no outneighbor, then we must have $C(u)=\sr_L$.
\end{property}

\subsection{Almost all \br{} graphs are far from acyclic} \label{sec:far-from-acyclic}

It is clear from the construction that no red vertex can participate in a cycle, but intuitively the $\sb \rightarrow\sb$ edges will result in many cycles
  in the blue part of the graph. ~\Cref{hehelemma} below makes this intuition precise. 

\begin{lemma}[$\BR$-graphs are far from acyclic]\label{hehelemma}
Let $d \geq \MinDegree$ be a constant. Then a random digraph $\bG\sim \BR$ is $\MaxEpsilon$-far from acyclic with probability $1-o_N(1)$.
\end{lemma}
\begin{proof}
It suffices to show that for any fixed coloring $C$,
  the random graph $\bG$ drawn using the same procedure running on $C$
  is far from acyclic with high probability.
To this end, we assume without loss of generality that the blue
  vertices in $C$ are $[N]$. We focus on the subgraph of $\bG$
  induced by the blue vertices $[N]$, which we refer to as the \emph{blue subgraph}.

The following claim is folklore and we include its proof for completeness:

\begin{claim}
An $N$-vertex digraph $G=(V,E)$ is $\eps$-far from acyclic if and only if for every \emph{(}bijective\emph{)} vertex ordering $\pi: V \to \{1,\dotsc,N\}$, the number of ``backedges'' \emph{(}i.e. directed edges $(u,v)$ such that $\pi(u) > \pi(v)$\emph{)} is at least $\eps dN$. 
\label{claim-ordering}
\end{claim}
\begin{proof}
We prove the contrapositive in both directions: ($\Rightarrow$) Deleting all the backedges leaves an acyclic graph, showing that the graph is $\eps$-close to acyclic. ($\Leftarrow$) Given a feedback arc set, after deleting it we can find a topological sort of the resulting acyclic graph. The ordering resulting from the topological sort has exactly the feedback arc set as its backedges. 
\end{proof}

We will use the following claim, which follows trivially from \Cref{claim-ordering}, to bound the distance to acyclicity
  of the blue subgraph of $\bG$: 

\begin{claim}\label{hahaclaim}
Let $G=(V,E)$ be an $N$-vertex digraph.
Suppose that for all balanced partitions $(V_1,V_2)$ of $V$,
  the number of directed edges from $V_1$ to $V_2$ is at least $\eps d N$,
  then $G$ is $\eps$-far from acyclic.
\end{claim}
\begin{proof}
Every ordering of vertices $\pi$ induces a balanced partition $(V_1,V_2)$
  by taking $V_2$ as the first $N/2$ vertices in $\pi$ and $V_1$ as
  the last $N/2$ vertices in $\pi$.
Then all edges from $V_1$ to $V_2$ are backedges with respect to $\pi$. The result follows from \Cref{claim-ordering}.
\end{proof}

Fix a balanced partition $(V_1,V_2)$ of the blue vertices $[N]$.
We show below that the number of edges from $V_1$ to $V_2$ in $\bG$ is
  at least $dN/20$ with probability $1-\exp({- N})$.
It follows from a union bound over all balanced partitions
  that with  probability $1-o_N(1)$, the number of edges one needs to delete to make $\bG$
  acyclic is at least $dN/20$.

To bound the number of edges in $\bG$ from $V_1$ to $V_2$,
  we go through vertices in $V_1$ one by one and for each
  vertex $u\in V_1$, draw a sequence of $d$ outneighbors without 
  replacement from a set of $2N-1$ vertices which contains $V_2$.
For each of these $dN/2$ many rounds and for any outcomes in previous rounds,
  the probability of gaining a directed edge from $V_1$ to $V_2$
  is at least
$$
\frac{(N/2)-(d-1)}{2N-1-(d-1)}
\ge \frac{1}{5}
$$
when $N$ is sufficiently large, so the expected number of edges is at least $(dN/2)\cdot (1/5)=dN/10$. It follows from a Chernoff bound (and a standard coupling argument) that the probability of having fewer than $dN/20$ edges is at most
$$\exp\big({-(dN/10)(1/2)^2(1/2)}\big) = \exp(-dN/80)\le \exp(- N),$$
when $d\ge 80$.
With a union bound over the at most $2^N$ many balanced partitions,
  we conclude~that $\bG$ has at least $dN/20$ edges from $V_1$ to $V_2$ in all balanced
  partitions $(V_1,V_2)$ of $[N]$ with probability at least
  $1-\exp(-N)\cdot 2^N=1-o_N(1)$.
Thus the blue subgraph is $(1/20)$-far from acyclic with probability $1-o_N(1)$. Because the total number of vertices is $3N$, after lifting back to the original graph we have that $\bG$ is $(1/60)$-far from acyclic with probability $1-o_N(1)$.
\end{proof}  

\section{Epochs and color revelation}\label{epochssec}

The goal of the rest of the paper is to prove \Cref{secondmaintheorem}.
Recall that under the vertex query model,
  each time an algorithm queries a vertex $u\in [3N]$
  it receives as its answer an ordered list $a=(v_1,\ldots,v_\ell)$ containing the outneighbors of $u$. We assume without loss of generality that the algorithm never queries the same vertex twice.
For \br{} graphs in the support of $\BR$ we know that the answer to each query is either 
  an ordered list $(v_1,\ldots,v_d)$ of $d$ distinct vertices different from $u$
  or the empty list.
This leads to the following definition of \emph{query histories}.
  
\begin{defn}[Query histories]
A \emph{query history} $H$ is an ordered tuple $((u_1,a_1),\ldots,(u_q,a_q))$
  for some $q\ge 0$ 
  such that $u_1,\ldots,u_q$ are \emph{distinct} vertices in $[3N]$ and each $a_i$ is 
  either a list of $d$ distinct vertices different from $u_i$ 
  or the empty list.
We refer to $q$ as the length of $H$,
  and $H$ as the empty history when $q=0$.
\end{defn}  

Each query history $H=((u_1,a_1),\ldots,(u_q,a_q))$
uniquely determines a \emph{knowledge graph}, denoted $\KG(H)$, which summarizes the information about the underlying graph contained in $H$: The vertex set of $\KG(H)$, denoted 
  $\VKG(H)$, consists of all vertices that appear in $H$ (i.e., every $u_i$ and every vertex
  $v$ in $a_i$ for some $i\in [q]$);
$\KG(H)$ contains a directed edge $(u,v)$ if $u=u_i$ and $v$ appears in $a_i$ for some $i\in [q]$.
Note that each vertex in $\KG(H)$ has outdegree either $d$ or $0$, and 
  every vertex with outdegree $d$ is queried in $H$.
On the other hand, a vertex $u$ with outdegree $0$
  has two cases: Either $u$ is queried in $H$ and the answer $a$ is empty,
  in which case we refer to $u$ as a \emph{sink} in $\KG(H)$,
  or $u$ is discovered as an outneighbor of some vertex queried in $H$
  but itself is never queried in $H$.

To prove \Cref{secondmaintheorem}, we introduce the notion of \emph{epochs}  
  and a new query model called the \emph{color revelation model} in \Cref{sec:model}.
In addition to receiving the adjacency list of the vertex queried,
  an algorithm under the color revelation model receives additional 
  information about colors of vertices in the current knowledge graph 
  at the end of each epoch. In the rest of the paper we show that, under the color revelation model,
  any $Q^*$-query deterministic algorithm finds a cycle in $\bG\sim \BR$ with probability  $o_N(1)$ (see the exact statement in \Cref{thirdmaintheorem}).
\Cref{secondmaintheorem} follows from \Cref{thirdmaintheorem} trivially because 
  the color revelation model is no harder than the vertex query model: any algorithm
  under the vertex query model
  can be simulated under the color revelation model by simply ignoring the additional information.

\subsection{The color revelation model}\label{sec:model}

Let $H=((u_1,a_1),\ldots,(u_q,a_q))$ be a query history for some $q\ge 0$; we write $H_i$ to denote its $i$-prefix $((u_1,a_1),\ldots,(u_i,a_i))$. We say the $k^{th}$ query $(u_k,a_k)$ is a \emph{surprise} in $H$ if $a_k$ contains a vertex that appears in $\VKG(H_{k-1})$. Otherwise, we refer to $(u_k,a_k)$ as \emph{surprise-free}.
 
We now describe the color revelation model, which provides additional power to the query algorithm by revealing the colors of vertices in previous epochs for free. Although this augmentation makes the task of cycle-finding easier, it also makes it easier to prove lower bounds. Formally, the oracle now contains a pair $(C,G)$ in the support of $\BR^*$, instead of just a \br{} graph $G$ as in the vertex-query model. The oracle uses $C$ to reveal to the algorithm colors of certain vertices. (In general, a coloring $C$ is not uniquely determined by a \br{} graph $G$.)  
  
Under the color revelation model, an algorithm $\alga$ maintains a triple $(H,\calE,P)$,
  where
\begin{flushleft}\begin{enumerate}
    \item $H$ is the current query history, updated after each query as in the vertex-query model;
  \item $\calE=(E_1,\ldots,E_\ell)$ for some $\ell\ge 1$ is a decomposition of $H$ into \emph{epochs}, where
  each epoch $E_i$ is by itself a query history and $H=E_1\circ \cdots \circ E_\ell$; and
  \item Letting  
  $H'=E_1\circ\cdots \circ E_{\ell-1}$, $P$ is a coloring map from $\VKG(H')$ to $\{\sb, \sr_1,\ldots,\sr_L\}$.
 \end{enumerate}\end{flushleft} 

Initially, $H$ and $E_1$ are empty and $\calE = (E_1)$. We refer to the final epoch $E_l$ as the \emph{current} epoch. For clarity, we use the symbols $P$ and $S$ to denote
  partial colorings over subsets of $[3N]$ and use $C$ to denote a full coloring
  over the vertex set $[3N]$.

Let $(H,\calE, P)$ denote the current triple maintained by an algorithm $\alga$. Under the color revelation model, the next round proceeds as follows:
\begin{flushleft}\begin{enumerate}
    \item As in the vertex query model, $\alga$ queries a vertex $u$, receives an ordered list $a$ containing the outneighbors of $u$ in $G$, and concatenates $(u,a)$ to $H$ and $E_\ell$.
    \item The current epoch \emph{ends} if $(u,a)$ is a surprise in $H$
    or $|E_\ell| = L/2$. In this case:
 \begin{enumerate}
       \item $\alga$ learns the colors of the vertices in the current epoch: $P$ is extended so that $P(u)=C(u)$ for every $u\in \VKG(H)$.
  \item A new epoch begins: An empty epoch $E_{\ell+1}$ is appended to $\calE$. 
   \end{enumerate}
\end{enumerate}\end{flushleft}

$\calE$ is can be reconstructed from $H$ by reading $H$ serially and recording the end of an epoch if a surprise occurs or the length of the epoch reaches $L/2$. Thus $\alga$ needs only to maintain the pair $(H,P)$ instead of the triple $(H,\calE,P)$. We refer to
  $\calE$ as the \emph{epoch decomposition} of $H$.

Next we introduce the notion of \emph{valid knowledge pairs}. 

\begin{defn}[Valid knowledge pairs]
A pair $(H,P)$ is called a \emph{valid knowledge pair} if
\begin{flushleft}\begin{itemize}
\item $H=((u_1,a_1),\ldots,(u_q,a_q))$ is a query
  history for~some $q\ge 0$ and $P$ is a coloring map over $\VKG(H')$, where $\calE=(E_1,\ldots,E_\ell)$
  is the epoch decomposition of $H$ and $H'=E_1\circ \cdots \circ E_{\ell-1}$;
\item There exists a pair $(C,G)$ in the support of $\BR^*$ such that 
  $C$ is an extension of $P$ and 
  $G$ is \emph{consistent} with $H$, i.e., $a_i$ is the adjacency list of $u_i$
  in $G$ for every $i\in [q]$.
\end{itemize}\end{flushleft}

 Given a valid knowledge 
   pair $(H,P)$ we use $\BR^*(H,P)$ to denote the distribution of 
  $(\bC,\bG)\sim\BR^*$ conditioning on $\bC$ being an extension of $P$
  and $\bG$ being consistent with $H$.
\end{defn}

Note that the pair $(H,P)$ maintained by an algorithm under the color revelation model is always valid by definition. From now on we consider a deterministic query algorithm $\alga$ under the color revelation model as a map from valid knowledge pairs to vertices so that $u=A(H,P)$ is the next vertex that is queried. \Cref{secondmaintheorem} follows directly from the following statement in the color revelation model:
  
\begin{restatable}{theorem}{mainmain}
 \label{thirdmaintheorem}
Let $d$ be a constant with $d\geq\MinDegree$, and let $\alga$ be a  $Q^* $-query deterministic algorithm that works on pairs in the support of $\BR ^*$ under the color revelation model, where $Q^* =N^{5/9}/\log N$. Then the probability of $\alga$ finding a cycle in $(\bC,\bG)\sim \BR ^*$ is $o_N(1)$.
\end{restatable}

\subsection{Epoch bounds}\label{sec:epoch-bound}

Let $(H,P)$ be a valid knowledge pair and let  
  $\calE=(E_1,\ldots,E_\ell)$ be the epoch decomposition of the query history $H$.
We refer to an epoch $E_i$, $i<\ell$, as a \emph{surprise epoch}
  if its last query is a surprise in $H$;
  otherwise $E_i$ has length $L/2$ and ends 
  by timeout.
A surprise epoch $E_i$ is a \emph{blue surprise epoch} if the last 
  vertex queried in $E_i$ is blue in $P$.
  
We begin our proof of \Cref{thirdmaintheorem} by proving upper bounds on the number of epochs and blue surprise 
  epochs that occur during the execution of a $Q$-query algorithm under the color revelation model.

\begin{lemma}[Epoch bound] \label{lemma:epochbound}
    There exists a constant $c_1$ such that for any algorithm that makes $Q$ queries, we have
    \begin{equation}
        \Prx_{(\bC,\bG) \sim \BR^*}\left[\text{~more than~} c_1\left(\frac{Q^2}{W} + \frac{Q}{L}\right) \text{~epochs occur~}\right] \leq \exp\left(-\Omega\left(\frac{Q^2}{W}\right)\right).
    \end{equation}
\end{lemma}

\begin{proof}
Let $\alga$ be an algorithm that makes $Q$ queries.
Since each epoch is either a surprise epoch or ends by timeout, the number of epochs which take place in running $\alga$ on a pair $(C,G)$ in the
  support of $\BR^*$ is bounded from above by the number of surprise queries 
  plus $2Q/L$.
As a result, it suffices to show that the probability of $\alga$ observing
  more than $O(Q^2/W)$ many surprises when running on $(\bC,\bG)\sim \BR^*$
  is at most $\exp(-\Omega(Q^2/W))$.

For this purpose we fix a valid knowledge pair $(H,P)$ and let $u=\alga(H,P)$ be the vertex that $\alga$ queries next. 
Below we upper bound the probability of $u$ being a surprise by $O(Q/W)$
  when $\alga$ runs on $(\bC,\bG)\sim \BR^*(H,P)$.
Since $u$ has not been queried before,
  a key observation is that, fixing any coloring $C$ 
  in the support of $\BR^*(H,P)$ and conditioning 
  $(\bC,\bG)\sim \BR^*(H,P)$ further on $\bC=C$,
  the adjacency list of $u$ is distributed as follows:
If $C(u)=\sb$ then each of its $d$ outneighbors is drawn without
  replacement from vertices of color $\sb$  in $C$ (other than $u$ itself) and
  vertices~of color $\sr_i$, $i\leq L/2$;
If $C(u)=\sr_i$ for some $i<L$ then each of its $d$ outneighbors is drawn 
  without replacement from vertices of color $\sr_{i+1}$ in $C$;
If $C(u)=\sr_{L}$, then its adjacency list is empty.
  
As a result, if $C(u)=\sb$, the probability of $u$ being a surprise query
  (as $(\bC,\bG)\sim \BR^*(H,P)$ further conditioning on $\bC=C$)
  is at most 
$$
d\cdot \frac{Q(d+1)}{2N-1}\le \frac{ d^2 Q}{N},
$$
using a union bound and the fact that $\VKG(H)$ has size at most 
  $q(d+1)\le Q(d+1)$. Similarly the probability of $u$ being a surprise when $C(u)=\sr_i$ for some $i<L$
  can be bounded from above by $2d^2Q/W$.
Since $u$ is always surprise-free if $C(u)=\sr_L$, we have that 
  the probability of $u$ being a surprise when $\alga$ runs on $(\bC,\bG)\sim \BR^*(H,P)$
  is at most $2d^2Q/W$.

    Now for each $q \in [Q]$, let $\bX_q$ be a Bernoulli random variable which is $1$ 
    if the $q^{th}$ query made by $\alga$ on $(\bC,\bG)\sim \BR^*$ is a surprise.
    Then what we have shown above implies that the probability of $\bX_q=1$ is
    $O(Q/W)$ even conditioning on any outcomes of $\bX_1,\ldots,\bX_{q-1}$.
    It then follows from the Chernoff bound  (together
with a standard coupling argument) that 
$$
\Pr\left[\sum_{q\in [Q]}\bX_q\ge \frac{4d^2Q^2}{W}\right]\le 
  \exp\left(-\Omega\left(\frac{Q^2}{W}\right)\right).
$$
This finishes the proof of the lemma.
\end{proof}

Recall that an epoch ends as a blue surprise epoch if the last query $u$ is both a 
  surprise and a blue vertex.
If we let $\bX_q$ denote the random variable that is $1$ if the $q^{th}$ query of $\alga$
  turns out to be the last query of a blue surprise epoch, when running on $(\bC,\bG)\sim \BR^*$, then the argument used in the proof of \Cref{lemma:epochbound} 
  implies that the probability of $\bX_q=1$ is at most $O(Q/N)$
  conditioning on any outcomes of $\bX_1,\ldots,\bX_{q-1}$.
This gives us the following upper bound:

\begin{lemma}[Blue surprise epochs bound] \label{lemma:blueepochbound}
    There exists a constant $c_2$ such that for any algorithm that makes $Q$ queries, we have
    \begin{equation}
        \Prx_{(\bC,\bG) \sim \BR^*}
        \left[\text{~more than~} \frac{c_2Q^2}{N}\ \text{blue surprise epochs occur~}\right]\le 
        \exp\left(-\Omega\left(\frac{Q^2}{N}\right)\right).
    \end{equation}
\end{lemma}

\section{Bounding the probability of long blue paths} \label{sec:no-long-blue-path}

In this section we prove a key lemma necessary for the proof of \Cref{thirdmaintheorem}: that in any given epoch, the probability that a $Q^*$-query algorithm discovers a ``long" path of previously unseen blue vertices is low. As a result, with high probability, the subgraph induced by the blue nodes revealed at the end of each epoch is a forest in which every tree has small depth.

\begin{lemma}[Long blue paths are unlikely.]\label{maintechnicallemma}
Let $(H,P)$ be a valid knowledge pair in which the length $q$
  of $H$ is bounded by $Q^*$.
Let $E_\ell$ be the current
  epoch of $H$ and let $(\bC,\bG)\sim \BR^*(H,P)$.
The probability that $\KG(E_\ell)$~contains a path of length 
  at least $4\log N$ consisting of blue vertices only under $\bC$ is $o(N^{-2})$.\footnote{The constant factors in the lemma statement are arbitrary but will be convenient later in the analysis. }
\end{lemma}

We begin with some notation and a sketch of the proof. Let $(H,P)$ be a valid knowledge pair, let $\calE=(E_1,\ldots,E_\ell)$ be the
  epoch decomposition of $H$, and let $H'=E_1\circ\cdots \circ E_{\ell-1}$.
By definition, the current epoch $E_\ell$ satisfies $|E_\ell|<L/2$ and every query in $E_\ell$ is surprise-free in $H$. As a result, the graph $\KG(E_\ell)$ is a vertex-disjoint union of 
  degree-$d$ out-trees.\footnote{Note that an~isolated vertex is also counted
  as a degree-$d$ out-tree.} This leads to the following observation:

\begin{property}\label{obse1}
Let $T$ be an out-tree in $\KG(E_\ell)$. 
\begin{enumerate}
    \item Every vertex of $T$, other than the root, lies outside $\VKG(H')$. (The root may or may not lie in $\VKG(H')$.)
    \item Every internal vertex and every sink in $T$ is queried in $E_\ell$.
\end{enumerate}
\end{property}

Let $\bS\sim \CH$ 
  be the distribution of partial colorings over $\VKG(H)$
  induced by $\bC$ drawn as in $(\bC,\bG)\sim \BR^*(H,P)$.
Then every partial coloring $S$ in the support of $\CH$ must be a 
  \emph{good} partial coloring over $\VKG(H)$ (see \Cref{basicproperty}):

\begin{defn}
We say $S$ is a \emph{good} partial coloring over $\VKG(H)$ if (1)
    $S$ is an extension of $P$, (2) For each directed edge $(u,v)$ in $\KG(H)$, either 
      $S(u)=S(v)=\sb,$ or $S(u)=\sb$ and $ S(v)=\sr_i$ for some $i\in [L/2]$,
      or $S(u)=\sr_i$ and $S(v)=\sr_{i+1}$ for some $i<L$, and (3)
      $S(u)=\sr_L$ for every sink vertex in $H$.
\end{defn}

In other words,  \Cref{maintechnicallemma} states that 
  $\KG(E_\ell)$ is unlikely to have a long blue path under $\bS\sim \CH$. 
To prove  \Cref{maintechnicallemma}, we introduce a \emph{naive 
  distribution} $\CH'$ in \Cref{naivedist} that is much easier to work with and at the 
  same time serves as a good approximation of the distribution $\CH$.
We then show that $\KG(E_\ell)$ is unlikely to have a long blue path under
  $\bS'\sim \CH'$, from which \Cref{maintechnicallemma} follows.

The intuition behind the naive distribution $\CH'$ is that we color each tree $T$ in $\KG(E_\ell)$~independently, ignoring all information in the knowledge pair $(H,P)$ other than the tree $T$ itself. Roughly speaking, we generate a coloring for $T$ as follows. If the root of $T$ lies outside of $\VKG(H')$, we color it red with probability 2/3 and blue with probability 1/3 as if it were drawn uniformly at random from $[3N]$. If the root of $T$ lies inside $\VKG(H')$, its color is known. We then propagate down the tree in breadth-first order.
If the parent of a vertex $v$ was colored $\sb$, $v$ is colored $\sb$ with probability $1/2$ and $\sr_i$ with probability $1/L$ for each $i\in [L/2]$; if the parent of $v$ was colored $\sr_i$ then $v$ is colored $\sr_{i+1}$. $\CH'$ does not capture $\CH$ perfectly, but we show in \Cref{goodapprox} that they are pointwise very close to each other.

\subsection{The naive distribution}\label{naivedist}

Before introducing the naive distribution $\CH'$,
  we start by classifying trees of $\KG(E_\ell)$ into four types and 
  note that we can already deduce colors of certain vertices in any
  good coloring $S$ over $\VKG(H)$.
Let $T$ be an out-tree of $\KG(E_\ell)$ with height $\sh(T)$ and root vertex $r$:
\begin{flushleft}\begin{itemize}
    \item $T$ is a \emph{type-$1$} out-tree if $r \in \VKG(H')$
    (so the color of $r$ has already been revealed in $P$) and $P(r)=\sr_i$ for some
    $i\in [L]$. It follows from \Cref{basicproperty}
    that every valid coloring $S$ has $S(v)=\sr_{i+\ell}$
    for each vertex $v$ of depth $\ell$ in the tree.
    (Note that we must have $i+\sh(T)\le L$; otherwise the pair $(H,P)$ cannot be a valid
    knowledge pair.)
    \item  $T$ is a \emph{type-$2$} out-tree if $r \in \VKG(H')$ and $P(r)=\sb$.  Then \emph{none} of its leaves
    can be a sink; otherwise $(H,P)$ implies that there is a path from
    a $\sb$ vertex to a $\sr_L$ vertex of length at most 
    $\sh(T)\le |E_\ell|<L/2$, contradicting with the validity of $(H,P)$.
    \item $T$ is a \emph{type-$3$} out-tree if $r$ 
    is not in $\VKG(H')$ but $T$ contains at least one sink leaf $v^*$.
    Given that $\sh(T)< L/2$,
      it follows from \Cref{basicproperty} 
      that every good coloring $S$ satisfies
      $S(r)=\sr_{L-k}$, where $k$ is the depth of $v^*$ in $T$,
      and $S(v)=\sr_{L-k+\ell}$ for every vertex $v$ of depth $\ell$ in the tree.
    \item $T$ is a \emph{type-$4$} out-tree if
     $r$ is not in $\VKG(H')$ and none of its leaves  is a sink.
\end{itemize}\end{flushleft}

\begin{figure}[th!]
    \centering
    \includegraphics[width=15cm]{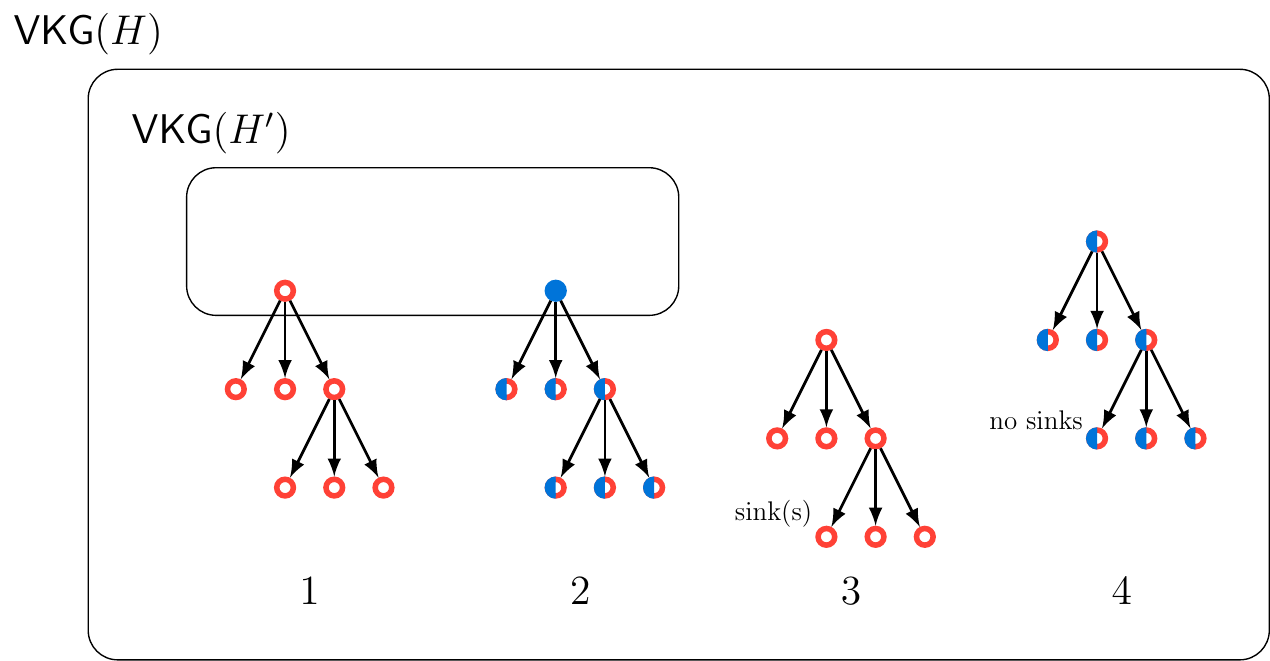}
    \caption{Possible out-tree types during a given epoch. Bichromatic circles denote vertices whose types cannot be determined from $(H,P)$ alone.}
    \label{fig:treetypes}
\end{figure}

\Cref{fig:treetypes} illustrates the four types of out-trees. Let $U$ denote the set of vertices that are always colored $\sr$ or always colored $\sb$ in a good partial coloring for $\VKG(H)$; that is, every vertex in $\VKG(H')$ as well as vertices in type-1 and type-3 trees in $\KG(E_l)$. Let $P'$ denote the unique partial coloring over $U$ that agrees with every good partial coloring $S$ over $\VKG(H)$. We let $Y := \VKG(H) \setminus U$ denote the set of vertices which may be colored either $\sr$ or $\sb$ in a good coloring; that is, all vertices in type-2 and type-4 trees except for the roots of type-2 trees.

We are ready to define the \emph{naive distribution} $\CH'$ of partial colorings over $\VKG(H)$.
A coloring $\bS'\sim \CH'$ is drawn using the following procedure:
\begin{flushleft}\begin{enumerate}
    \item First we color each vertex $u\in U$ as $P'(u)$ (so $\bS'$ is always
    an extension of $P'$);
    \item For each type-$4$ tree $T$, color its root vertex $r$ blue with probability 
    $N/(3N-\sh(T)W)$ and $\sr_i$ with probability $Y/(3N-\sh(T)W)$ for each $i\le L-\sh(T)$.
    (The intuition behind the denominator is that because there is a path of length $\sh(T)$ that starts
    at $r$, its color cannot be $\sr_{L-\sh(T)+1},\ldots,\sr_L$.)
\item Then we go through each type-$2$ and type-$4$ tree one by one and 
      consider uncolored vertices in breadth-first order.
    For each vertex $v$, if its parent is colored $\sb$, color $v$ with $\sb$ with probability $1/2$
    and with $\sr_i$ for each $i\in [L/2]$ with probability $1/L$.
    If the parent of $v$ is colored $\sr_i$, color $v$ $\sr_{i+1}$.\footnote{ Observe that we never go beyond $\sr_L$ because the height of each tree
      is at most $L/2$.}
\end{enumerate}\end{flushleft}
The following property follows directly from the procedure for $\CH'$ above:

\begin{property}
Every partial coloring in the support of $\CH'$ is a good 
  partial coloring over $\VKG(H)$.
\end{property}

Both distributions $\CH$ and $\CH'$ are supported on good partial colorings over $\VKG(H)$. 
The next lemma shows that $\CH'$ is  a good approximation of $\CH$:
\begin{lemma}\label{lemmaapprox}
For every good partial coloring $S$ over $\VKG(H)$, we have 
\begin{equation}\label{eq1}
0. 9 \cdot \Pr_{\bS'\sim \CH'} \big[\bS'=S\big]\le 
\Pr_{\bS\sim \CH} \big[\bS=S\big]\le 1. 1 \cdot 
\Pr_{\bS'\sim \CH'} \big[\bS'=S\big].
\end{equation}
\end{lemma}

Before proving \Cref{lemmaapprox} in \Cref{goodapprox}, 
  we use it to give a quick proof of \Cref{maintechnicallemma}.

\begin{proof}[Proof of \Cref{maintechnicallemma} using \Cref{lemmaapprox}]
Given a good coloring $S$ over $\VKG(H)$ 
  we use $\text{LBP}(S)$ to denote the event that $\KG(E_\ell)$ contains
  a blue path of length at least $4\log N$ under $S$.
It follows from \Cref{lemmaapprox} that
\begin{equation}\label{hehehaha}
\Pr_{\bS\sim \CH} \big[\text{LBP}(\bS)\big]\le 1.1\cdot \Pr_{\bS'\sim \CH'}
  \big[\text{LBP}(\bS')\big].
\end{equation}
On the other hand, if $\text{LBP}(S)$ holds then there must be a vertex $v$
  in either a type-$2$ or a type-$4$ tree (since every vertex in a type-$1$ or type-$3$ tree must be red in a good partial coloring)
  such that $v$ is of depth at least $4\log N$ and the path from the root
  to $v$ is all blue.
For each such vertex $v$,
  let $\text{LBP}(S,v)$ denote the event that the path from the root to
  $v$ is blue under $S$.
Then the probability of $\text{LBP}(\bS',v)$ when $\bS'\sim \CH'$
  is $(1/2)^{\ell}\le 1/N^4$ if $v$ is in a type-$2$ tree and has
  depth $\ell$, and is
$$
\frac{N}{3N-\sh(T)W}\cdot (1/2)^\ell\le 1/N^4
$$
if $v$ is in a type-$4$ tree $T$.
As a result, the probability of $\text{LBP}(\bS')$ when $\bS'\sim \CH'$
  is $O(1/N^3)$ by a union bound since the number of $v$ is trivially 
  at most $3N$. The lemma follows from \eqref{hehehaha}.
\end{proof} 

\subsection{The naive distribution is a good approximation: Proof of \Cref{lemmaapprox}}\label{goodapprox}

To simplify the presentation, in this section we use the notation ``$a\pm b$'' to denote a quantity that is between $a-b$ and $a+b$.

Let $S$ be a good partial coloring over $\VKG(H)$.
We write $\calT_2$ to denote the set of type-$2$ trees in $\KG(E_\ell)$ and $\calT_4$ to denote the set of type-$4$ trees in $\KG(E_\ell)$,
  and we write $\calT$ to denote $\calT_2\cup \calT_4$. 
Given $S$, we write $\calT_{4,b}(S)$ to denote the set of type-$4$ trees with a blue 
  root and $\calT_{4,r}(S)$ to denote the set of type-$4$ trees with a red root in $S$.
We also use $\#_{br}( S),\#_{bb}( S)$ and $\#_{rr}( S)$ 
  to denote the total number of blue-red, blue-blue and red-red edges in all trees in $\calT$.

We start with the easier task of obtaining a closed-form expression for 
  $\Pr_{\bS'\sim\CH'}[\bS'=S]$. 
This quantity can be written as a product: each root of a type-$4$ tree 
  contributes a factor which depends on its~color in $S$ (recall the second step of the procedure for drawing from $\CH'$), and each edge of a
  tree in $\calT$ contributes a factor which is $1/2$ if it is a blue-blue edge in $S$,
  $1/L$ if it is a blue-red edge, and $1$ if it is a red-red edge.
As a result, we have 
\begin{align*}
\Pr_{\bS'\sim\CH'}\big[\bS'=S\big]
&= \left( \prod_{T\in \calT_{4,b}(S)}\frac{N}{3N-\sh(T)W} \right)\left(
 \prod_{T\in \calT_{4,r}(S)}\frac{W}{3N-\sh(T)W}\right)  
\left(\frac{1}{L}\right)^{ \#_{br}( S)}
\left(\frac{1}{2}\right)^{ \#_{bb}( S)}\\
&= \left(\prod_{T\in \calT_4}\frac{W}{3N-\sh(T)W}\right)  
\left(\frac{L}{2}\right)^{|\calT_{4,b}(S)|}  \left(\frac{1}{L}\right)^{ \#_{br}( S)}
\left(\frac{1}{2}\right)^{ \#_{bb}(S)}\\
&= \tau_1\cdot \left(\frac{L}{2}\right)^{|\calT_{4,b}(S)|}\left(\frac{1}{L}\right)^{\#_{br}(S)}\left(\frac{1}{2}\right)^{\#_{bb}(S)},
\end{align*}
where the second equality uses $WL/2=N$ and the fact that ${\cal T}_4$ is the disjoint union of ${\cal T}_{4,b}(S)$ and ${\cal T}_{4,r}(S)$, and the quantity  $\tau_1>0$ is a value that does not depend on $S$.

Next we work on the probability distribution $\CH$.
For each good partial coloring $S$ over $\VKG(H)$, 
  we write $\sw(S)$ as a shorthand for
\begin{equation}\label{defw}
\sw(S):= \Pr_{(\bC,\bG)\sim \BR^*}\big[\text{$\bC$ is an extension of $S$ and
  $\bG$ is consistent with $H$}\big].
\end{equation}
Given the definition of $\CH$ and $\sw(\cdot)$, we have
\begin{equation}\label{standard}
\Pr_{\bS\sim \CH}\big[\bS=S\big] =
\frac{\sw(S)}{\sum_{\text{good}\ S'} \sw(S')},
\end{equation}
where the sum is over all good partial colorings
  $S'$ over $\VKG(H)$.

Looking ahead, our plan is to show that there is a value
  $\tau>0$ (independent of $\bS$) such that
\begin{equation}\label{keyineq}
0.99\cdot \tau\cdot \Pr_{\bS'\sim \CH'}\big[\bS'=S\big]\le 
\sw(S)\le 1.01 \cdot \tau\cdot \Pr_{\bS'\sim \CH'}\big[\bS'=S\big].
\end{equation}
With (\ref{keyineq}), it follows from $\sum_{\text{good}\ S'}
  \Pr_{\bS'\sim \CH'}[\bS'=S']=1$ (which holds because $\CH'$ is supported on good colorings)
  that
$$
0.99\tau\le
\sum_{\text{good}\ S'}\sw(S')\le 1.01\tau.
$$
Combining this with (\ref{keyineq}) and (\ref{standard}), we have
$$
\Pr_{\bS\sim \CH}\big[\bS=S\big]\le \frac{1.01\cdot \tau\cdot \Pr_{\bS'\sim\CH'}
[\bS'=S]}{0.99\tau}<1.1\cdot \Pr_{\bS'\sim\CH'}\big[\bS'=S\big]
$$
and the other side of (\ref{eq1}) can be proved similarly.

So it suffices to prove (\ref{keyineq}). 
We start with some notation. Recall that 
   $\bU$ is the uniform distribution over all full colorings.
Given a full coloring $C$, we use $\BR(C)$ to denote the distribution
  of \br{} graphs generated using $C$ as the full coloring 
  in the procedure for $\BR$.
  
Now we consider the $(\bC,\bG)\sim \BR^*$ in the definition of $\sw(S)$ (recall (\ref{defw})) by first drawing  
  a full~coloring $\bC\sim \bU$.
If $\bC$ is not an extension of $S$ then we already fail to satisfy the 
  condition in the definition of $\sw(S)$.
If $\bC$ is an extension of $S$ then we draw $\bG\sim \BR(\bC)$ to
  see if $\bG$ is consistent with $H$.
  
A~useful~observation is that every $C$ that extends $S$ shares the same probability of $\bG\sim \BR(C)$ being consistent with~$H$. Let $\#_b(U)$ (respectively $\#_r(U)$)
  be the number of blue (respectively red) vertices in $U$ under $S$ that are \emph{queried}~in~$H$; note that these
  two numbers are independent of $S$ since every good  coloring must be an extension of $P'$ on $U$.
Let $\#_b(Y,S)$ (respectively $\#_r(Y,S)$) 
  denote the number of blue (respectively red) vertices in $Y$ under $S$
  that are \emph{queried} in $H$.
Then for every $C$ that is an extension of $S$,
  the probability of $\bG\sim \BR(C)$ being consistent with $H$ is
\begin{align}\nonumber
& \left(\frac{1}{(2N-1)\cdots (2N-d)}\right)^{\#_b(U)+\#_b(Y,S)}
\left(\frac{1}{W\cdots (W-d+1)}\right)^{\#_r(U)+\#_r(Y,S)}\\ 
&\hspace{1cm}=\tau_2\cdot \left(\frac{1}{(2N-1)\cdots (2N-d)}\right)^{ \#_b(Y,S)}
\left(\frac{1}{W\cdots (W-d+1)}\right)^{ \#_r(Y,S)},\label{hahahehehe}
\end{align}
for some positive value $\tau_2$ independent of $S$.
Note that our choices of $L,W$ and $Q^*$ satisfy
\begin{equation}\label{keykey}
LQ^*=o(W).
\end{equation}
Using (\ref{keykey}) (we only need $L=o(W)$ here) and the fact that $\#_b(Y,S),\#_r(Y,S)\le L/2$,
  (\ref{hahahehehe}) becomes
$$
(1\pm o_N(1))\cdot \tau_2\cdot \left(\frac{1}{2N}\right)^{d\cdot \#_b(Y,S)}\cdot 
\left(\frac{1}{W}\right)^{d\cdot \#_r(Y,S)}
=(1\pm o_N(1))\cdot \tau_3\cdot \left(\frac{1}{L}\right)^{d\cdot \#_b(Y,S)}  ,
$$
for some positive value $\tau_3$ that is independent of $S$ since $\#_b(Y,S)+\#_r(Y,S)$ is a constant
  independent of $S$.

Note that $d\cdot \#_b(Y,S)= \#_{bb}( S)+\#_{br}( S)-d |\calT_2|$.
This is just because each blue vertex queried in~$Y$ introduces $d$ edges that 
  are either blue-blue or blue-red in $\calT$;
  we need to subtract $d |\calT_2|$ because roots of type-$2$ trees
  are not included in $Y$.
Since $|\calT_2|$ is a value independent of $S$, (\ref{hahahehehe}) can be simplified to 
\begin{align*}
(1\pm o_N(1))\cdot \tau_4\cdot \left(\frac{1}{L}\right)^{\#_{bb}( S)+\#_{br}( S)},
\end{align*}
for some positive value $\tau_4$ that is independent of $S$.
As a result, we have 
\begin{equation} \label{eq:brine}
\sw(S)=(1\pm o_N(1))\cdot 
\Pr_{\bC\sim \bU}\big[\text{$\bC$ is an extension of $S$}\big]
\cdot \tau_4\cdot  \left(\frac{1}{L}\right)^{\#_{bb}( S)+\#_{br}( S)}.
\end{equation}

Next evaluate the probability that $\bC\sim \bU$ is an extension of $S$
  over $\VKG(H)=U\cup Y$.
For this purpose we consider the following experiment:
\begin{flushleft}\begin{enumerate}
    \item Pick an arbitrary ordering $u_1,\ldots,u_{|U|}$ of $U$ and an arbitrary
      ordering $y_1,\ldots,y_{|Y|}$ of $Y$.
    \item Start with $N$ $\sb$ pebbles and $W$ $\sr_i$ pebbles for each $i\in [L]$.
    Go through vertices $u_1,\ldots,u_{|U|}$ one by one and assign each one a remaining (as yet unassigned) pebble uniformly
    at random. Then go through vertices $y_1,\ldots,y_{|Y|}$ one by one and assign each one a remaining
    pebble uniformly at random.
    \item For each $u_i$, we use $\bX_i$ to denote the Bernoulli random variable that is $1$
    if $u_i$ is assigned a pebble of color $S(u_i)$, and define $\bY_i$ similarly for each $y_i$.
\end{enumerate}\end{flushleft}
Then the probability $\Pr_{\bC\sim \bU}[\bC$ is an extension of $S]$ that we are interested in is 
\begin{align*}
 \Pr\big[\bX_1= \cdots=\bY_{|Y|}=1\big] 
& =\Pr\big[\bX_1=\cdots =\bX_{|U|}=1]\cdot \prod_{i\in [|Y|]}
\Pr\big[\bY_i=1\hspace{0.06cm}|\hspace{0.06cm}\bX_1=\cdots=\bY_{i-1}=1\big]\\
& =\tau_5\cdot \prod_{i\in [|Y|]}
\Pr\big[\bY_i=1\hspace{0.06cm}|\hspace{0.06cm}\bX_1=\cdots=\bY_{i-1}=1\big],
\end{align*}
for some positive value $\tau_5$ that is independent of $S$.
For each $y_i$ with $S(y_i)=\sb$, we have 
$$
\frac{N-Q^*(d+1)}{3N}\le 
\Pr\big[\bY_i=1\hspace{0.06cm}|\hspace{0.06cm}\bX_1=\cdots=\bY_{i-1}=1\big]\le 
\frac{N}{3N-Q^*(d+1)}.
$$
This is because regardless of outcomes for vertices before $y_i$,
  the number of blue pebbles left in the round of $y_i$ lies between $N-Q^*(d+1)$ and $N$ (since $\VKG(H)$
  has no more than $Q^*(d+1)$ vertices) and the total number of pebbles
  left is between $3N-Q^*(d+1)$ and $3N$.

Similarly for each $y_i$ with $S(y_i)=\sr_j$ for some $j\in [L]$, we have 
$$
\frac{W-Q^*(d+1)}{3N}\le 
\Pr\big[\bY_i=1\hspace{0.06cm}|\hspace{0.06cm}\bX_1=\cdots=\bY_{i-1}=1\big]\le 
\frac{W}{3N-Q^*(d+1)}
$$
Let $\#^*_b(Y,S)$ (or $\#^*_r(Y,S)$) denote the number of blue (or red) vertices in $Y$ under $S$
 (unlike $\#_b(Y,S)$ and $\#_r(Y,S)$, these vertices may have not been queried).
 It follows from (\ref{keykey}) and  
  $|Y|=O(L)$ that
\begin{align}
 \Pr_{\bC\sim \bU}\big[\text{$\bC$ is an extension of $S$}\big] 
 &=(1\pm o_N(1))\cdot \tau_5\cdot \left(\frac{1}{3}\right)^{\#^*_b(Y,S)}
\left(\frac{W}{3N}\right)^{\#^*_r(Y,S)}\nonumber \\[0.5ex]
 &=(1\pm o_N(1))\cdot \tau_6\cdot  \left(\frac{2}{L}\right)^{\#^*_r(Y,S)}, \label{eq:cucumber}
\end{align}
for some value $\tau_6>0$ that is independent of $S$ since
  $\#_b^*(Y,S)+\#_r^*(Y,S)$ is a constant independent of $S$. 
Finally we have  
$$
\frac{\sw(S)}{\Pr_{\bS'\sim \CH'}[\bS'=S]}
=(1\pm o_N(1)) \cdot \tau_7 \cdot \left(\frac{2}{L}\right)^{
\#_r^*(Y,S)+\#_{bb}(S)+|\calT_{4,b}(S)|}
$$
for some value $\tau_7>0$ that is independent of $S$.
Note that for any good coloring $S$, the quantity 
  $\#_r^*(Y,S)+\#_{bb}(S)+|\calT_{4,b}(S)|$ is equal to $|Y|$, a constant that 
  does not depend on $S$. This finishes the proof of (\ref{keyineq}) and \Cref{lemmaapprox}.

\section{A Lower Bound on Cycle Finding}
\label{sec:cycle-finding}

This section combines the results of Lemmas \ref{lemma:epochbound}, \ref{lemma:blueepochbound} and \ref{maintechnicallemma} to 
  establish \Cref{thirdmaintheorem}, which is restated below:

% restate main theorem
\mainmain*

\begin{proof}
Let $\calA$ be a $Q^*$-query algorithm.
We start with the definition of \emph{typical} pairs in the support of $\BR^*$
  with respect to $\calA$, and then show that
  $(\bC,\bG)\sim\BR^*$ is typical with probability $1-o_N(1)$.

\begin{defn} \label{def:typical}
We say a pair $(C,G)$ in the support of $\BR^*$
  is \emph{typical} with respect to an algorithm
  $\calA$ if the following conditions hold:
\begin{flushleft}\begin{enumerate}
    \item[(i)] The number of epochs during the execution of $\calA$ on 
      $(C,G)$ is $O(\UBQ^2/W + \UBQ/L)$.\vspace{-0.05cm}
    \item[(ii)] The number of blue surprise epochs during the execution of $\calA$
      on $(C,G)$ is $\smash{O(\UBQ^2/N)}$.\vspace{-0.05cm}
    \item[(iii)] For each $q\in [\UBQ]$, let $(H^{(q)},P^{(q)})$ denote the 
      knowledge pair of running $\calA$ on $(C,G)$ after\\ $q$ queries and let
      $E^{(q)}$ denote the current epoch \emph{(}in the epoch decomposition of 
      $H^{(q)}$\emph{)}.
    Then there is no blue path longer than $4\log N$ in $\KG(E^{(q)})$
      under the coloring $C$.
\end{enumerate}\end{flushleft}
\end{defn}

We combine Lemmas \ref{lemma:epochbound}, \ref{lemma:blueepochbound} and \ref{maintechnicallemma} to show that $(\bC,\bG)\sim \BR^*$ is typical 
  with respect to $\calA$
  with probability $1-o_N(1)$.
We focus on the third condition (iii) since the probability of $(\bC,\bG)$ satisfying
  the first two conditions is $1-o_N(1)$ by Lemmas \ref{lemma:epochbound} and \ref{lemma:blueepochbound}.
For (iii) we have
$$
\Pr_{(\bC,\bG)\sim \BR^*}\big[\text{$(\bC,\bG)$ violates (iii)}\big]
\le \sum_{q=1}^{Q^*} \Pr_{(\bC,\bG)\sim \BR^*}\big[\text{$(\bC,\bG)$
  violates (iii) after $q$ queries}\big].
$$
On the other hand, the $q^{th}$ probability in the sum can be written as
\begin{align}\nonumber
\sum_{\text{valid}\ (H,P)} &\Pr_{(\bC,\bG)\sim\BR^*}\big[\text{$\calA$ observes
  $(H,P)$ after $q$ queries on $(\bC,\bG)$}\big]\\
  &\hspace{0.8cm}\times \Pr_{(\bC,\bG)\sim \BR^*(H,P)}\big[\text{$(\bC,\bG)$ violates (iii) 
  after $q$ queries}\big],\label{qthprob}
\end{align}
where the sum is over all valid knowledge pairs $(H,P)$ of length $q$.
It follows from \Cref{maintechnicallemma} that the latter 
  probability in $(\ref{qthprob})$ is $o(N^{-2})$ for every valid knowledge
  pair $(H,P)$.
As a result, the probability of $(\bC,\bG)\sim\BR^*$ violating
  (iii) is $o(N^{-1})$ and thus $(\bC,\bG)$ is typical with probability $1-o_N(1)$.

Given a query history $H$ and a vertex $u$, we write $\anc(H,u)$
  to denote the set of ancestors of $u$ in $\KG(H)$, i.e., the set
  of vertices (other than $u$ itself) that have a directed path to $u$.
(If $u\notin \VKG(H)$ then $\anc(H,u)$ is trivially empty.)
The claim below shows that if $(C,G)$ is typical then at any time 
  during the execution of $\calA$ on $(C,G)$, every blue vertex 
  has a \emph{small} set of ancestors in $\KG(H)$. 

\begin{claim}\label{lowanc}
Let $(C,G)$ be a typical pair with respect to $\calA$.
Then for each $q\in [Q^*]$, letting $H$ be the query history of $\calA$
  after making $q$ queries on $(C,G)$,
  we have 
\begin{equation}\label{boundancestor}
\big|\anc(H,u)\big|\le O\left(\log N\cdot \left(\frac{{Q^*}^2}{W}+\frac{Q^\ast}{L}\right)\cdot \frac{{Q^*}^2}{N}\right), 
\end{equation}
for every vertex $u$ with $C(u)=\sb$.
\end{claim}
\begin{proof}
Recall that at the end of each blue surprise epoch, $\calA$ may find an edge $(u,v)$ such that the vertex $u$ being queried is blue and $v$ is a vertex encountered before. We refer to such an edge as a \emph{surprise edge} if $v$ also turns out to be blue.

Now we consider running $\calA$ on a typical pair $(C,G)$.
Let $(H^{(i)},P^{(i)})$ 
  denote the knowledge pair maintained by $\calA$ after $i$ queries,
  let $E^{(i)}$ be the current epoch, and let $H^{(i)}=H'^{(i)}\circ E^{(i)}$.
We focus on the evolution of the blue subgraph (the subgraph
  induced by its blue vertices) of 
  $\KG(H'^{(i)})$ over time.
We write $\BKG(H'^{(i)})$ to denote the blue subgraph of $\KG(H'^{(i)})$.

First we note that $\KG(H'^{(i)})$ (and thus, $\BKG(H'^{(i)})$) 
  is only updated at the end of each epoch. 
If an epoch ends at the $i^{th}$ query,
  a number of out-trees are added to $\KG(H'^{
  (i-1)})$. 
Each such tree (other than its root) is vertex-disjoint from $\KG(H'^{(i-1)})$.
In addition, if the epoch is a blue surprise epoch, no more than $d$ many
  surprise edges are added to $\KG(H'^{(i)})$.
Now focusing on $\BKG(H'^{(i)})$ vs $\BKG(H'^{(i-1)})$, 
  we have that at the end of each epoch, each out-tree added
  to $\BKG(H'^{(i-1)})$ satisfies the extra condition of having depth at most $4\log N$.
If it is the end of a blue surprise epoch we may need to add no more than $d$ surprise edges to
  $\BKG(H'^{(i)})$ .
  
As a result, letting $H$ be the query history of $\calA$ after making
  $q$ queries on a typical pair $(C,G)$, 
  we have that $\BKG(H')$ 
  is the union of (1) a forest in which each out-tree has depth at most
\begin{equation}\label{eq111}
O\left(\log N\cdot \left( \frac{{Q^*}^2}{W}+\frac{Q^\ast}{L}\right)\right)
\end{equation}
and (2) a set of at most
\begin{equation}\label{eq222}
O\left(\frac{{Q^*}^2}{N}\right)
\end{equation}
many surprise edges, where (\ref{eq111}) follows from the bound for the number of 
  epochs and (\ref{eq222}) follows found the bound for the number of
  blue surprise epochs given that $(C,G)$ is typical.
  
Let $u$ be a vertex in $\BKG(H')$.
To bound the number of its ancestors, we 
  consider an in-tree $T$ rooted at 
  $u$ such that every ancestor of $u$ appears in $T$ (with a directed path to $u$).
If we remove surprise edges from $T$, it is left with a vertex-disjoint
  union of directed paths; this is because, after removing surprise edges,
  $\BKG(H')$ is a forest of out-trees (so no vertex has indegree larger than $1$).
Since each path has length bounded by (\ref{eq111}) and the number of
  surprise edges is bounded by (\ref{eq222}),
  the number of vertices in $T$ (or the number of 
  ancestors of $u$) is bounded by (\ref{boundancestor}).

Now in \Cref{lowanc}, $u$ can be a vertex in $\VKG(E)$.
Note that $\KG(E)$ must be a forest of out-trees and because $(C,G)$ is typical,
  the blue subgraph of each such out-tree has depth at most $4\log N$.
As a result, considering vertices in $\VKG(E)$ may add a term 
  of $4\log N$ to our bound for the number of ancestors, which is still
  captured by (\ref{boundancestor}).
This finishes the proof of the claim.
\end{proof}

Now we show that $\calA$ finds a cycle in $(\bC,\bG)\sim \BR^*$ with
  probability $o_N(1)$. 
Given that $(\bC,\bG)$ is typical with probability at least $1-o_N(1)$, we have
\begin{align}
\Pr_{(\bC,\bG)\sim\BR^*}\big[\text{$\calA$ finds a cycle}\big]\label{maineq}\\
&\hspace{-3cm}\le o_N(1) + \Pr_{(\bC,\bG)\sim\BR^*}\big[\text{$(\bC,\bG)$ is typical
  and $\calA$ finds a cycle}\big]\nonumber \\[0.5ex]
&\hspace{-3cm}\le o_N(1) +\sum_{q\in [Q^*]}\Pr_{(\bC,\bG)\sim\BR^*}\big[\text{$(\bC,\bG)$ is typical
  and $\calA$ finds a cycle in the $q^{th}$ round}\big].\nonumber
\end{align}
As a result, it suffices to bound each probability in the sum by $o(1/N)$.

Fix any $q\in [Q^*]$.
Given a pair $(H,C)$, where $H=H'\circ E$ is a query history of length $q-1$
  and $C$ is a full coloring,
  we write $P(H,C)$ to denote the restriction of $C$ on $\VKG(H')$.
Then we have
\begin{align}
&\Pr_{(\bC,\bG)\sim\BR^*}\big[\text{$(\bC,\bG)$ is typical
  and $\calA$ finds a cycle in the $q^{th}$ round}\big] \label{hehe1eq}\\
&\hspace{2cm}\le \sum_{(H,C)} \Pr_{(\bC,\bG)\sim \BR^*}\big[\text{$\bC=C$
  and $\calA$ running
  on $(\bC,\bG)$ observes $(H,P(H,C))$}\big]\nonumber\\
  &\hspace{4cm}\times \Pr_{(\bC,\bG)\sim \BR^*(H,C)}
\big[\text{$\calA$ finds a cycle in the $q^{th}$ round}\big],\label{hehe2eq}
\end{align}
where the sum is over all $(H,C)$ such that every blue vertex
  (under $C$) in $\VKG(H)$ has its number of ancestors bounded by (\ref{boundancestor}).
This follows from  \Cref{lowanc} since $(\bC,\bG)$ is typical in (\ref{hehe1eq}).
  
Fix such a pair $(H,C)$ and let $u=\calA(H,P(H,C))$ be the 
  next vertex that is queried by $\calA$.
If $u\notin \VKG(H)$ or $u\in \VKG(H)$ is not blue under $C$,
  the probability in (\ref{hehe2eq}) is trivially $0$.
If $u\in \VKG(H)$ is blue,~we note that under $(\bC,\bG)\sim \BR^*(H,C)$, outneighbors of
  $u$ are picked randomly from $2N-1$ vertices without replacement. Consequently
 the probability that one of them is an ancestor of $u$ can be bounded by
$$
O\left(\log N \cdot \left(\frac{\UBQ^4}{WN^2} + \frac{\UBQ^3}{LN^2}\right)\right).
$$
As a result, this is an upper bound for (\ref{hehe2eq})
  as well as (\ref{hehe1eq}) and thus,
  the sum in (\ref{maineq}) is at most
\begin{equation}\label{hehe3eq}
        \UBQ \cdot O\left(\log N\cdot \left(\frac{\UBQ^4}{WN^2} + \frac{\UBQ^3}{LN^2}\right)\right)
        = O\left(\log N \cdot \left(\frac{\UBQ^5}{WN^2} + \frac{\UBQ^4}{LN^2}\right)\right)=o_N(1) 
    \end{equation}
with our choices of $L,W$ and $Q^*$.
This finishes the proof of the lemma. 
\end{proof}

Looking back regarding our choices of $L:=(2N)^{2/9}$ and $W := 2N/L = (2N)^{7/9}$, 
  we need $L,W$ and $Q^*$ to satisfy the following two inequalities
  for the proof to work: 
  (\ref{keykey}): $LQ^*=o(W)$ and that (\ref{hehe3eq}) above
  is $o_N(1)$.
Our choices of $L$ and $W$ are optimized to maximize the query
  complexity $Q^*$ under these conditions.

\section{Finding cycles in \br-graphs using $O(N^{13/18})$ queries}
\label{sec:sketch-algorithm}

Given the lower bound established above for cycle-finding in \br{}  graphs, one natural question concerns the limitations of this approach: what is the true query complexity of cycle-finding in graphs drawn from this distribution? This section sketches two algorithmic approaches that find cycles with high probability in random graphs $\bG \sim \BR$ for many values of the length parameter $L$. In particular, setting $L=\Theta(N^{2/9})$ as in our lower bound construction yields an algorithm for cycle finding in $\BR$ graphs with query complexity roughly $N^{13/18}$.

\textbf{Algorithm 1.} We begin with the following simple observation:  With high probability over a random \br~graph $\bG \sim \BR$, for each vertex $v \in [3N]$ it is possible to correctly determine the color (and layer, if the color is red) of $v$ in $O(L)$ queries.  This is a straightforward consequence of the following two facts. First, if $v$ is a red vertex in layer $R_i$, then every directed path from $v$ reaches a sink after exactly $L-i$ edges. Second, for almost every graph $\bG \in \BR$, a sequence of random walks made from any blue vertex in $\bG$ will differ significantly in the distance they travel before they find a sink. Thus an algorithm can determine the color and layer of $v$ with high probability by making several random walks of length $O(L)$.

We can leverage this observation to find a cycle with high probability in $O(L\sqrt{N})$ queries as follows. The algorithm works by first identifying a blue vertex in $O(L)$ queries by randomly sampling and confirming its color using the procedure described above. Each child of a blue vertex $\bG \sim \BR$ is blue with probability 1/2, so we can grow a blue path from our seed vertex at a cost of roughly $O(L)$ queries to confirm the color of each additional vertex.\footnote{ With high probability, the number of red vertices we identify is proportional to the length of the path. If a blue node has no blue children, an event which occurs with probability $1/2^d$, we backtrack to the previous node.} We construct a blue path of length $C \sqrt{N}$, at which point each successive blue vertex added to the path creates a cycle with probability at least $C/(2 \sqrt{N})$. By a birthday paradox argument, the next $C \sqrt{N}$ blue vertices added to the path yield a cycle with high probability for large $C$. Setting $L = (2N)^{2/9}$ as in our lower bound proof, the query complexity of the resulting algorithm is $O(N^{13/18})$.

{\bf Algorithm 2.} Algorithm 1 provides a good upper bound on the query complexity of cycle-finding in \br{}  graphs when $L$ is relatively small. In this section we sketch a more sophisticated strategy that gives a query-efficient algorithm when $L$ is large.
The key observation here is that for almost every graph $\bG \sim \BR,$ given any red vertex $v$ in layer $R_i$, by making $P := \tilde{O}(W)$ queries it is possible to query almost every vertex in layer $R_{i+t}$, where $t=\log_d P.$ This is accomplished by performing a breadth-first search of depth $t$ starting from vertex $v$.  We think of this as the query algorithm ``building a wall'' at layer $R_{i+t}$.

Algorithm 2 has two stages. In the first stage, it builds a series of walls, effectively mapping out the structure of $G$. In the second stage, it exploits its knowledge about the structure of $G$ to efficiently build a long blue path using a method similar to Algorithm 1. 

In more detail, the first stage starts by first identifying $M$ red vertices by sampling vertices at random and using random walks to confirm their color, a process which takes $O(LM)$ queries. In the rest of the first stage the algorithm then performs the wall-building procedure at each of these vertices, a process which takes roughly $\tilde{O}(MW)$ queries.\footnote{ If the algorithm finds a sink while attempting to run a breadth-first search of depth $t$, this wall fails and the procedure continues. } At this point, the query algorithm has built $\Theta(M)$ walls, which with high probability are typically spaced at intervals roughly $O(L/M)$ apart throughout the layers $R_1,\dots,R_L$. Thus the first stage takes about $\tilde{O}(M(W+L)) = \tilde{O}(M(N/L+L))$ queries.

In the second stage, Algorithm 2 is the same as Algorithm 1, except that Algorithm 2 can identify vertex colors by reaching a wall instead of a sink vertex. Consider a random walk from a vertex $v$. If $v$ is in layer $R_i$, then most random walks from $v$ will collide with the next wall in a particular, fixed number of queries $a_i$, which will typically be $O(L/M)$. If $v$ is a blue vertex, then most random walks from $v$ will still collide with a wall in $O(L/M)$ queries, but with high probability the length of these walks will vary significantly. As a result, using the same method as Algorithm 1, Algorithm 2 can identify vertex colors in $O(L/M)$ queries and find a cycle of length $O(\sqrt{N})$ in $O(L\sqrt{N}/M)$ queries. Thus the query complexity of Algorithm 2 is $\tilde{O}(M(N/L+L) + L\sqrt{N}/M)$. If $L \gg N^{1/4}$, then taking $M={\frac {N^{1/4}L}{\sqrt{N + L^2}}} \gg 1$, we get that the query complexity of this second approach is roughly $\tilde{O}(N^{1/4} \sqrt{N + L^2})$, which is $o(N)$ for $N^{1/4} \ll L \leq o(N^{3/4}).$

\section{Directions for future work:  towards upper bounds} \label{sec:future}

Given our $\tilde{\Omega}(N^{5/9})$ lower bound, it is natural to ask the true query complexity of cycle finding in sparse digraphs that are $\eps$-far from acyclic. We conjecture that there is an $o(N)$-query algorithm for this problem, and we pose the problem of finding such an algorithm as a tantalizing goal for future work.  We conclude with a few comments towards this goal:

\begin{flushleft}\begin{enumerate}
\item  Let $\ell = \ell(m,\eps)$ be the smallest value such that every $m$-edge digraph $G$ with the smallest feedback arc set of size at least $\eps m$
  must have a cycle of length at most $\ell.$  Fox \cite{Fox18} has proved that $\ell(m,\eps) \le
{\tilde{O}(\log m)}/{\eps}$.
It follows that every bounded-outdegree-$d$ $N$-vertex digraph that is constant-far from acyclic must contain a cycle of length $\tilde{O}(\log N)$.  This structural result may be viewed as a highly efficient \emph{nondeterministic} algorithm (with\\ query complexity $\tilde{O}(\log N)$) for the cycle-finding problem that we consider.

\item It is possible that a simple algorithm based on breadth first search may have sublinear query complexity for cycle-finding in far-from-acyclic bounded-degree digraphs.  In more detail, we do not know a counterexample to the following conjecture:  ``Let $0 < \eps < 1$ be a (small) constant.  Let $\calA'$ be an algorithm which works as follows:  for $C=C(\eps)$ (a large constant) many repetitions, $\calA'$ picks a random vertex $v$ in $G$ and performs a breadth first search out from $v$ until $CN/\log N$ vertices have been explored.  When run on any $N$-vertex graph $G$ that is $\eps$-far from acyclic, one of the $C(\eps)$ BFS's performed by algorithm $\calA'$ finds a cycle with constant probability.'' (We note that by considering the case in which $G$ is a union of $d$ many randomly chosen bipartite matchings, it can be shown that $N/\log N$ cannot be replaced by any function of $N$ that is $o(N/\log N)$.)
\end{enumerate}\end{flushleft}

\section*{Acknowledgements}  

We thank Jacob Fox for telling us about \cite{Fox18}.  X.C. is supported by NSF IIS-1838154 and NSF CCF-1703925.
R.A.S. is supported by NSF IIS-1838154, NSF CCF-1814873, NSF CCF-1563155, and by the Simons Collaboration on Algorithms and Geometry.

\bibliographystyle{alpha}
\begin{flushleft}
\bibliography{allrefs}
\end{flushleft}

\end{document}